\documentclass[12pt,english]{article}
\usepackage[utf8]{inputenc}
\usepackage[T1]{fontenc}
\usepackage{geometry}
\usepackage{babel}
\usepackage{amsthm}
\usepackage{amsmath}
\usepackage{amssymb}
\usepackage{esint}
\usepackage[bookmarks=false,colorlinks,citecolor=blue,urlcolor=blue]{hyperref}
\usepackage[titletoc,toc,title]{appendix}

\makeatletter
\numberwithin{equation}{section}
\numberwithin{figure}{section}
\theoremstyle{plain}
\newtheorem{thm}{\protect\theoremname}
  \theoremstyle{definition}
  \newtheorem{defn}[thm]{\protect\definitionname}
  \theoremstyle{plain}
  \newtheorem{prop}[thm]{\protect\propositionname}
  \theoremstyle{plain}
  \newtheorem{lem}[thm]{\protect\lemmaname}

\date{}

\makeatother

  \providecommand{\definitionname}{Definition}
  \providecommand{\lemmaname}{Lemma}
  \providecommand{\propositionname}{Proposition}
\providecommand{\theoremname}{Theorem}

\newcommand{\res}[1]{
\underset{#1}{\mathrm{Res}}
} 

\begin{document}

\title{Non-commutative integration, zeta functions and the Haar state for
$SU_{q}(2)$}

\author{Marco Matassa%
\thanks{SISSA, Via Bonomea 265, I-34136 Trieste, Italy. \textit{E-mail address}:
marco.matassa@gmail.com%
}}

\maketitle
We study a notion of non-commutative integration, in the spirit of modular spectral triples, for the quantum group $SU_{q}(2)$. In particular we define the non-commutative integral as the residue at the spectral dimension of a zeta function, which is constructed using a Dirac operator and a weight. We consider the Dirac operator introduced by Kaad and Senior and a family of weights depending on two parameters, which
are related to the diagonal automorphisms of $SU_{q}(2)$. We show that, after fixing one of the parameters, the non-commutative integral
coincides with the Haar state of $SU_{q}(2)$. Moreover we can impose an additional condition on the zeta function, which also fixes the
second parameter. For this unique choice the spectral dimension coincides with the classical dimension.

\section{Introduction}
Many works have been devoted to studying how quantum groups and their homogeneous spaces fit into the framework of spectral triples \cite{con-book}. In the last years, in particular, there have been several proposals of extensions of this notion, which could be used to accomodate these classes of non-commutative geometries. This is quite natural, since the axioms of a spectral triple are tailored on the commutative situation, and one naturally expects that new features, which appear only in the non-commutative world, should play a role in this description. We mention two among such proposals of extensions: the first one is the framework of \textit{twisted spectral
triples} \cite{type III}, which modifies the commutator condition by introducing a twist on the algebra; the second one is that of \textit{modular spectral triples} \cite{modular1,modular2,modular3} which, roughly speaking, modifies the resolvent condition by replacing the operator trace with a weight. Notice that these approaches modify two different axioms of a spectral triple. Also, more recently, there has been an
attempt to merge these two approaches \cite{kaad}, which was motivated
by a construction for the quantum group $SU_{q}(2)$ \cite{suq2-kaad}.

There is a natural question that arises by considering the framework of modular spectral triples: if we are allowed to replace the operator trace by a weight, is there any preferred choice?
Let us consider this question for the case of compact quantum groups, where we can be guided by symmetry. In this setting it is well known
that there is a unique state, the Haar state, which is the non-commutative analogue of the Haar integral for compact groups. But the choice of a state gives a notion of non-commutative integration, as it is known
from the theory of von Neumann algebras \cite{tak}.
Therefore we could ask for the following condition for a spectral triple: the non-commutative integral, which is defined in terms of the Dirac operator, should coincide with the Haar state.
However it is clear that this is not possible in the usual setting: indeed it follows, from
the properties of spectral triples, that the non-commutative integral is a trace, while the Haar state in general is not.
On the other hand in the extended frameworks we mentioned above, in particular for modular spectral triples, the non-commutative integral need not be a trace. Therefore such a requirement can be in principle satisfied. Moreover, getting back
to the question of a preferred choice, one could use such a criterion
to choose a natural weight in the context of modular spectral triples.

Here we will investigate in some detail this question for the case of the quantum group $SU_{q}(2)$. More specifically we consider the Dirac operator $D_{q}$ introduced in \cite{suq2-kaad}, which gives a (twisted) modular spectral triple. Our non-commutative integral will be defined as the residue at the spectral dimension of a certain zeta function. This zeta function is constructed in a natural way from the operator $D_{q}$ and a weight $\phi$. We consider a family of weights, depending on two parameters $a,b\in\mathbb{R}$, which essentially parametrize the most general diagonal automorphism on $SU_{q}(2)$. Our first
task will be to determine for which values of the parameters the zeta function is well defined, and determine its spectral dimension.

At this point we can impose our requirement of recovering the Haar state from the non-commutative integral. A necessary condition
is that their modular groups coincide.
We will show that this condition fixes $b=1$, but leaves $a$ undetermined.
Moreover, we can indeed show that the non-commutative integral, once properly normalized, recovers the Haar state independently of
the value of $a$. So up to this point our program is partially successful, since we still have a considerable freedom in choosing $a$.
On the other hand the spectral dimension $n$ depends on $a$, in particular $n = a + 1$.
We could therefore fix $a = 2$ in such a way that $n$ coincides with the classical dimension.
It turns out that this choice also arises by considering an additional requirement on the zeta function, as we now proceed to explain.

Indeed so far we have only used the information which is contained in the residue at $z=n$, that is at the spectral dimension. But the
analytic continuation of the zeta function contains much more information, as it is known from the heat kernel expansion in the classical case. In this respect the residue at $z=n$ corresponds only to the first coefficient of this expansion. Then we can look at the next non-trivial coefficient, which from the point of view of the zeta function consists in looking at another residue. It is easy to see that this coefficient vanishes for the operator obtained in the classical limit from $D_{q}$. Therefore we can require an analogue condition for the non-commutative case. It turns out that this condition is satisfied only in the case $a=2$, which was the natural choice we mentioned above.

The structure of the paper is as follows. In Section \ref{sec:integration} we give the relevant definitions for the non-commutative integral and prove an easy theorem that determines its modular group.
In Section \ref{sec:SU_q(2)} we provide the necessary background on the quantum group $SU_{q}(2)$.
In Section \ref{sec:Dirac} we describe the Dirac operator $D_{q}$ and point out some additional properties it satisfies.
In Section \ref{sec:zeta} we study in detail the family of zeta functions defined in terms of $D_{q}$ and the family of weights, which depends on the parameters $a, b \in \mathbb{R}$. We also prove that fixing the parameter $b=1$ is a necessary condition to recover
the Haar state.
In Section \ref{sec:Haar} we prove that the non-commutative integral coincides with the Haar state, independently of the value of the parameter $a$.
Finally in Section \ref{sec:scalar} we introduce an additional requirement, which is shown to be satisfied only for $a = 2$, therefore fixing completely the choice of the weight.

\section{Non-commutative integration}
\label{sec:integration}

In the description of non-commutative geometry via spectral triples a fundamental role is played by the Dirac operator $D$, which is
used to formulate many geometrical notions at the level of operator algebras. For example the notion of integration, which is our main
focus in this paper, can be expressed as the Dixmier trace of a certain
power of this operator. Indeed, in the case of a compact manifold
of dimension $n$, for any continuous function $f$ the Dixmier trace
of $f|D|^{-n}$ coincides with the usual integral, up to a normalization
constant. However computing the Dixmier trace is not an easy task
in general, and for this reason it is useful to reformulate this notion
of integration in a way which is easier to handle.

One such reformulation is achieved by defining the non-commutative
integral as a residue of a zeta function involving $D$, as is done
for example in the case of the local index formula \cite{local-index}.
Going back to the manifold case, we can define this zeta function
as $\zeta_{f}(z)=\mathrm{Tr}(f|D|^{-z})$, for $z\in\mathbb{C}$.
Then $\zeta_{f}(z)$ is holomorphic for all $\mathrm{Re}(z)>n$ and
the residue at $z=n$ coincides with the integral of $f$, again up
to a normalization constant. For general results that relate the Dixmier trace, the asymptotics of the zeta function and the heat kernel expansion see \cite{asym-zeta} (see also \cite{int-loc} for the non-unital case). The results of these papers are moreover valid in the semifinite setting, which is going to be of interest for us.

The aim of this section is to define the non-commutative integral
as the residue at the spectral dimension \cite{spec-dim} of a certain
zeta function. However, differently from the usual case, our zeta
function is not defined solely in terms of the spectrum of $D$, but
involves also the choice of a weight $\phi$. The usual setting is
recovered by taking $\phi$ equal to the operator trace. The motivation
comes from the approach of modular spectral triples, which, as we
mentioned in the introduction, allows the choice of such a weight.
Although we have in mind the specific example of $SU_{q}(2)$, the content of this section is more general in nature.

\subsection{The weight}
The two definitions of non-commutative integration that we mentioned, that is the Dixmier trace and the residue of the zeta function, rely crucially on the behaviour of the spectrum of $D$.
But in the non-commutative world, especially in the case of quantum groups, the
spectra of naturally defined operators can be very different from
their commutative counterparts. This usually spoils summability conditions,
or even the compactness condition. One way out of these problems is
to consider a Dirac operator which has the same spectrum as the classical
one, which is the idea of isospectral deformations \cite{connes-landi}, see also \cite{netu10}.
However it is clear that in this case, by remaining in the realm of
usual spectral triples, one obtains a non-commutative integral which
has the trace property.
Therefore, if we consider the case of quantum groups, we do not recover the Haar state, since it is a non-tracial state.

Another possibility is to modify the definition of spectral triple
to account for such features. One framework which goes in this direction
is that of modular spectral triples \cite{modular1,modular2,modular3}, see also \cite{kaad},
where the idea is to define the notions of compactness and summability with respect to a weight $\phi$. We now briefly recall this setting, since it is essentially the one we will have in mind for the rest
of the paper.

Let $\mathcal{N}$ be a semifinite von Neumann algebra.
We fix a faithful normal semifinite weight $\phi:\mathcal{N}_{+}\to[0,\infty]$,
with modular group $\sigma^{\phi}$. Moreover we require that
$\phi$ is strictly semifinite, which implies that $\phi$ descends
to a faithful normal semifinite trace on $\mathcal{N}^{\sigma^{\phi}}$,
the fixed point algebra. We say that an operator is compact with respect
to $\phi$ if it belongs to $\mathcal{K}(\mathcal{N}^{\sigma^{\phi}},\phi)$,
which is the smallest norm closed $*$-ideal in $\mathcal{N}^{\sigma^{\phi}}$
containing the projections $P\in\mathcal{N}^{\sigma^{\phi}}$ such
that $\phi(P)<\infty$. With these notations we can give the definition.
\begin{defn}
We call the triple $(\mathcal{A},\mathcal{H},D)$ a unital \textit{modular
$\sigma$-spectral triple}, with respect to the pair $(\mathcal{N},\phi)$,
it if satisfies the following conditions:
\begin{enumerate}
\item $\mathcal{A}$ is a unital $*$-subalgebra of $\mathcal{N}$, which
acts on a separable Hilbert space $\mathcal{H}$,
\item $D$ is a self-adjoint operator affiliated with the fixed point algebra
$\mathcal{N}^{\sigma^{\phi}}$,
\item $\mathcal{A}$ consists of analytic vectors for the modular group
$\sigma^{\phi}$,
\item $\sigma:\mathcal{A}\to\mathcal{A}$ is an algebra automorphism,
\item $[D,x]_{\sigma}=Dx-\sigma(x)D$ extends to a bounded operator for
each $x\in\mathcal{A}$, 
\item $(D^{2}+1)^{-1/2}$ is compact with respect to the weight $\phi$.
\end{enumerate}
\end{defn}
We formulate this definition as it appears in \cite{kaad}, which also
considers the possibility of a twist in the commutator. This twist,
which we denote by $\sigma$, is a priori unrelated to the modular
group $\sigma^{\phi}$. In the original papers, which also provide
some examples, the twist does not appear. An example which uses the
twist in the commutator is that of $SU_{q}(2)$ given in \cite{suq2-kaad},
which we will consider in some detail later. Other examples in this spirit, which however being non-unital do not fit in this scheme, are given in \cite{kmink1,kmink2}.

For our purposes, that is for the definition of a notion of non-commutative integration, it is sufficient to consider a stripped-down version of this framework.
We consider a $*$-algebra
$\mathcal{A}$, which is represented as bounded operators on a Hilbert
space $\mathcal{H}$, and an unbounded self-adjoint operator $D$ acting this
space. We take our weight $\phi$ to be of the form $\mathrm{Tr}(\Delta_{\phi}\cdot)$,
where $\Delta_{\phi}$ is a positive and invertible operator. This
is essentially the statement of the Radon–Nikodym theorem for semifinite
weights on von Neumann algebras.

First of all we define a zeta function in terms of $D$ and $\phi$,
and the corresponding notion of spectral dimension.
In the following we will provide definitions which are appropriate to the case of compact spaces.
We also assume for simplicity that $D$ is invertible.

\begin{defn}
The \textit{zeta function} associated to $D$ and $\phi$ is defined by
$$\zeta(z):=\phi(|D|^{-z})=\mathrm{Tr}(\Delta_{\phi}|D|^{-z}).$$
If it exists, we define the \textit{spectral dimension} to be the number
$$n:=\inf\{s>0:\zeta(s)<\infty\}.$$
\end{defn}
In the following we will assume that $\zeta(z)$ has a simple pole
at $z=n$. This condition is related to the ideals $\mathcal{Z}_{p}$
introduced in \cite{asym-zeta}.

We now want to extend the zeta function to $\mathcal{A}$, which is going to be the domain of our non-commutative integral.
As a compatibility requirement between $\mathcal{A}$ and
$\phi$ we ask that $\sigma^{\phi}(x)\in\mathcal{A}$, where $\sigma^{\phi}(x) = \Delta_{\phi}^{-1} x \Delta_{\phi}$
is the modular group of $\phi$. Notice that this condition is part
of the requirements for a modular spectral triple.

\begin{defn}
For any $x\in\mathcal{A}$ we also define the zeta function
$$\zeta_{x}(z):=\phi(x|D|^{-z})=\mathrm{Tr}(\Delta_{\phi}x|D|^{-z}).$$
The \textit{non-commutative integral} is the linear functional $\varphi: \mathcal{A} \to \mathbb{C}$ defined by
$$\varphi(x) := \res{z=n}\ \zeta_{x}(z) = \res{z=n}\ \mathrm{Tr}(\Delta_{\phi}x|D|^{-z}),$$
where $n$ is the spectral dimension.
\end{defn}

We remark that, thanks to the condition $\sigma^{\phi}(x)\in\mathcal{A}$,
the zeta function $\zeta_{x}(z)$ exists for $\mathrm{Re}(z)>n$ and
has at most a simple pole. Indeed $\Delta_{\phi}x|D|^{-z}=\sigma^{\phi}(x)\Delta_{\phi}|D|^{-z}$
and the statement easily follows by applying Hölder's inequality for
the trace.

\subsection{The modular group}
We now want to determine the modular group of the non-commutative
integral $\varphi$. By this we mean the automorphism $\theta$ such
that $\varphi(xy)=\varphi(\theta(y)x)$ for $x,y \in \mathcal{A}$. It is clear that $\theta$ is trivial in the commutative case, simply because the pointwise product of functions is commutative.
This is also true, although not trivially, if $D$ satisfies the conditions for a spectral triple.

On the other hand, by considering the setting of twisted spectral triples, where we require the twisted commutator to be bounded for some automorphism $\sigma$, we find under suitable conditions that $\varphi(xy)=\varphi(\sigma^{n}(y)x)$.
An additional twisting appears, in the setting of modular spectral triples, via the modular operator $\Delta_{\phi}$ associated to a weight $\phi$.

Here we want to take both these twistings into account and prove a theorem that, under some simplifying assumptions, determines the modular group of the non-commutative integral.
We assume that $[D,x]_{\sigma}=Dx-\sigma(x)D$ is bounded for every $x \in \mathcal{A}$, for a fixed automorphism $\sigma$, and that it satisfies a regularity property specified below.
We also consider an algebra $\mathcal{A}$ which is defined in terms of generators and relations, as for compact quantum groups, and require that $\sigma$ acts diagonally  on its generators.
These conditions can be clearly weakened, but they suffice for the examples that we have in mind.

\begin{thm}
\label{thm:modular}
With the assumptions above, consider the non-commutative integral $\varphi$ with spectral dimension $n$. Assume furthermore that $D$ satisfies the following regularity property:
\begin{itemize}
\item for some fixed $0<r\leq1$, we have that $|D|^{r}[|D|^{s},x]_{\sigma^{s}}|D|^{-s}$ is a bounded operator for every element $x\in\mathcal{A}$ and for all $s\geq n$.
\end{itemize}
Then the modular group of $\varphi$ is given by $\theta=\sigma^{\phi}\circ\sigma^{n}$.
\end{thm}
\begin{proof}
To prove the result we show that the following chain of equalities
holds
\begin{equation*}
\begin{split}
\varphi(xy) & = \res{s=n}\ \mathrm{Tr}(\Delta_{\phi}xy|D|^{-s})\\
 & = \res{s=n}\ \mathrm{Tr}(\Delta_{\phi}x|D|^{-s}\sigma^{s}(y))\\
 & = \res{s=n}\ \mathrm{Tr}(\Delta_{\phi}x|D|^{-s}\sigma^{n}(y))\\
 & = \res{s=n}\ \mathrm{Tr}(\Delta_{\phi}\sigma^{\phi}(\sigma^{n}(y))x|D|^{-s})=\varphi(\theta(y)x).
\end{split}
\end{equation*}
Let us start with the first one. Consider the expression
$$\mathrm{Tr}(\Delta_{\phi}xy|D|^{-s})-\mathrm{Tr}(\Delta_{\phi}x|D|^{-s}\sigma^{s}(y)) = \mathrm{Tr}(\Delta_{\phi}x|D|^{-s}[|D|^{s},y]_{\sigma^{s}}|D|^{-s}),$$
where we have used the identity
$$
y|D|^{-s}-|D|^{-s}\sigma^{s}(y)=|D|^{-s}[|D|^{s},y]_{\sigma^{s}}|D|^{-s}.
$$
Now consider $0<r\leq1$ as in the statement. Using Hölder's inequality we have
\begin{equation*}
\begin{split}
\left|\mathrm{Tr}\left(\Delta_{\phi}x|D|^{-s}[|D|^{s},y]_{\sigma^{s}}|D|^{-s}\right)\right| & =\left|\mathrm{Tr}\left(\Delta_{\phi}x|D|^{-(s+r)}|D|^{r}[|D|^{s},y]_{\sigma^{s}}|D|^{-s}\right)\right|\\
 & \leq\mathrm{Tr}\left(|\Delta_{\phi}x|D|^{-(s+r)}|\right)\left\Vert |D|^{r}[|D|^{s},y]_{\sigma^{s}}|D|^{-s}\right\Vert.
\end{split}
\end{equation*}
By assumption $\||D|^{r}[|D|^{s},y]_{\sigma^{s}}|D|^{-s}\|$ is finite
for all $s>n$. Therefore, since
$$\res{s=n}\ \mathrm{Tr}(\Delta_{\phi}x|D|^{-(s+r)})=0,$$
we find that the residue at $s=n$ of this term vanishes, hence we have
$$\varphi(xy)= \res{s=n}\ \mathrm{Tr}(\Delta_{\phi}xy|D|^{-s}) = \res{s=n}\ \mathrm{Tr}(\Delta_{\phi}x|D|^{-s}\sigma^{s}(y)).$$

The second step consists in proving that
$$\res{s=n}\ \mathrm{Tr}(\Delta_{\phi}x|D|^{-s}\sigma^{s}(y))
= \res{s=n}\  \mathrm{Tr}(\Delta_{\phi}x|D|^{-s}\sigma^{n}(y)).$$
Subtracting these quantities and using Hölder's inequality we find
$$\left|\mathrm{Tr}\left(\Delta_{\phi}x|D|^{-s}(\sigma^{s}(y)-\sigma^{n}(y))\right)\right|\leq\mathrm{Tr}\left(|\Delta_{\phi}x|D|^{-s}|\right)\|\sigma^{s}(y)-\sigma^{n}(y)\|.$$
Since $\sigma$ acts diagonally on the generators of $\mathcal{A}$, the element $y \in \mathcal{A}$ can be written as a finite sum of homogeneous elements with respect to $\sigma$.
Therefore we can consider, without loss of generality, that $\sigma(y)=\lambda y$ for some $\lambda$. Therefore we have
$\|\sigma^{s}(y)-\sigma^{n}(y)\|=|\lambda^{s}-\lambda^{n}|\|y\|$.
The residue at $s=n$ of this quantity vanishes, so we get 
$$\res{s=n}\ \left|\mathrm{Tr}\left(\Delta_{\phi}x|D|^{-s}(\sigma^{s}(y)-\sigma^{n}(y)\right)\right|
\leq \res{s=n}\ \mathrm{Tr}\left(|\Delta_{\phi}x|D|^{-s}|\right)|\lambda^{s}-\lambda^{n}|\|y\|
= 0.$$

For the last step have to show that
$$\res{s=n}\ \mathrm{Tr}(\Delta_{\phi}x|D|^{-s}\sigma^{n}(y))
= \res{s=n}\ \mathrm{Tr}(\Delta_{\phi}\sigma^{\phi}(\sigma^{n}(y))x|D|^{-s}).$$
But this immediately follows from the trace property and the property
of the modular operator $\sigma^{n}(y)\Delta_{\phi}=\Delta_{\phi}\sigma^{\phi}(\sigma^{n}(y))$.
Finally, putting all the steps together and denoting $\theta=\sigma^{\phi}\circ\sigma^{n}$,
we have that $\varphi(xy)=\varphi(\theta(y)x)$. The proof is complete.
\end{proof}
Let us now consider the case in which $\mathcal{A}$ is the coordinate
algebra of a compact quantum group, with its Haar state $h$ having
the modular group $\vartheta$. We can use this theorem as a criterion
to check if the non-commutative integral $\varphi$ coincides with
the Haar state $h$. Indeed a necessary condition for this to happen
is that the modular group $\theta$ of $\varphi$ coincides with $\vartheta$.

The strategy is the following: given an operator $D$ and a weight
$\phi$ we consider the associated zeta function and compute its spectral
dimension which, if it exists, we denote by $n$. Then if $D$ satisfies
the conditions of Theorem \ref{thm:modular}, we check if the modular
group $\theta=\sigma^{\phi}\circ\sigma^{n}$ coincides with the modular
group $\vartheta$ of the Haar state. In the rest of the paper we
will perform this analysis for the case of the quantum group $SU_{q}(2)$,
using the Dirac operator $D_{q}$ introduced in \cite{suq2-kaad}.
We will also briefly mention the case of the Podle\'{s} sphere.

\section{Background on $SU_{q}(2)$}
\label{sec:SU_q(2)}

In this section we provide some background on the quantum group $SU_{q}(2)$,
which we will be our focus in the rest of the paper. We use the notations
and conventions of the book by Klimyk and Schmüdgen \cite{KS}. For
$0<q<1$ we denote by $\mathcal{A}:=\mathcal{O}(SU_{q}(2))$ the unital
Hopf $*$-algebra with generators $a,b,c,d$ satisfying the relations
\[
\begin{gathered}ab=qba\ ,\quad ac=qca\ ,\quad bd=qdb\ ,\quad cd=qdc\ ,\quad bc=cb\ ,\\
ad=1+qcb\ ,\quad da=1+q^{-1}bc\ ,
\end{gathered}
\]
with the usual Hopf algebra structure and the involution given by
\[
a^{*}=d\ ,\quad b^{*}=-qc\ ,\quad c^{*}=-q^{-1}b\ ,\quad d^{*}=a\ .
\]
For each $l\in\frac{1}{2}\mathbb{N}_{0}$ there is a unique (up to
unitary equivalence) irreducible corepresentation $V_{l}$ of the
coalgebra $\mathcal{A}$ of dimension $2l+1$. If we fix a vector
space basis in each $V_{l}$, and denote by $t_{i,j}^{l}\in\mathcal{A}$
the corresponding matrix coefficients, then we have the following
analogue of the Peter-Weyl theorem: the set $\{t_{i,j}^{l}\in\mathcal{A}:l\in\frac{1}{2}\mathbb{N}_{0},\ -l\leq i,j\leq l\}$
is a vector space basis of $\mathcal{A}$. With a suitable choice
of basis in $V_{1/2}$ one has
\[
a=t_{-1/2,-1/2}^{1/2}\ ,\quad b=t_{-1/2,1/2}^{1/2}\ ,\quad c=t_{1/2,-1/2}^{1/2}\ ,\quad d=t_{1/2,1/2}^{1/2}\ .
\]

We also need to consider the quantized enveloping algebra $U_{q}(\mathfrak{sl}(2))$.
This is a Hopf algebra generated by $k,k^{-1},e,f$ with relations
\[
kk^{-1}=k^{-1}k=1\ ,\quad ke=qek\ ,\quad kf=q^{-1}ek\ ,\quad[e,f]=\frac{k^{2}-k^{-2}}{q-q^{-1}}\ .
\]
It carries the Hopf algebra structure
\[
\begin{gathered}\Delta(k)=k\otimes k\ ,\quad\Delta(e)=e\otimes k+k^{-1}\otimes e\ ,\quad\Delta(f)=f\otimes k+k^{-1}\otimes f\ ,\\
S(k)=k^{-1}\ ,\quad S(e)=-qe\ ,\quad S(f)=-q^{-1}f\ ,\\
\varepsilon(k)=1\ ,\quad\varepsilon(e)=\varepsilon(f)=0\ .
\end{gathered}
\]
It becomes a Hopf $*$-algebra, which we denote by $U_{q}(\mathfrak{su}(2))$,
by adding the involution
\[
k^{*}=k\ ,\quad e^{*}=f\ ,\quad f^{*}=e\ .
\]
There is a dual pairing between the Hopf algebras $U_{q}(\mathfrak{sl}(2))$
and $\mathcal{A}$, which we denote by $\langle\cdot,\cdot\rangle$.
This pairing is used to define the left and right actions of $U_{q}(\mathfrak{sl}(2))$
on $\mathcal{A}$ by the formulae
\[
g\triangleright x:=x_{(1)}\langle g,x_{(2)}\rangle\ ,\quad g\triangleleft x:=\langle g,x_{(1)}\rangle x_{(2)}\ ,\qquad x\in\mathcal{A}\ ,\ g\in U_{q}(\mathfrak{sl}(2))\ ,
\]
where we used Sweedler's notation for the coproduct. These actions
make $\mathcal{A}$ into a $U_{q}(\mathfrak{sl}(2))$-bimodule. For
the actions of the generators on the basis $t_{i,j}^{l}$ we have
\[
\begin{gathered}k\triangleright t_{i,j}^{l}=q^{j}t_{i,j}^{l}\ ,\qquad t_{i,j}^{l}\triangleleft k=q^{i}t_{i,j}^{l}\ ,\\
e\triangleright t_{i,j}^{l}=\sqrt{[l+1/2]^{2}-[j+1/2]^{2}}t_{i,j+1}^{l}\ ,\quad f\triangleright t_{i,j}^{l}=\sqrt{[l+1/2]^{2}-[j-1/2]^{2}}t_{i,j-1}^{l}\ .
\end{gathered}
\]
In the previous formulae we have used the $q$-numbers $[x]_{q}$,
which are defined as $[x]_{q}:=(q^{-x}-q^{x})/(q^{-1}-q)$. In the
following we will also use the notation
\[
\partial_{k}:=k\triangleright\ ,\quad\partial_{e}:=e\triangleright\ ,\quad\partial_{f}:=f\triangleright\ ,
\]
especially as operators acting on a suitable completion of $\mathcal{A}$.
Observe that, since $\Delta(k^{n})=k^{n}\otimes k^{n}$ for every
$n\in\mathbb{Z}$, we have that $k^{n}\triangleright$ and $\triangleleft k^{n}$
are algebra automorphisms on $\mathcal{A}$.

\subsection{The Haar state}

We denote by $A:=C^{*}(SU_{q}(2))$ the universal $C^{*}$-completion
of the $*$-algebra $\mathcal{A}$. Let $h$ be the Haar state of
$A$ whose values on the basis elements are
\[
h(a^{i}b^{j}c^{k})=h(d^{i}b^{j}c^{k})=\delta_{i0}\delta_{jk}(-1)^{k}[k+1]^{-1}\ ,\quad h(t_{i,j}^{l})=\delta_{l0}\ .
\]
Let $\mathcal{H}_{h}$ denote the GNS space $L^{2}(A,h)$, where the
inner product is defined by $(x,y):=h(x^{*}y)$. The representation
of $A$ on $\mathcal{H}_{h}$ is induced by left multiplication in
$A$. The set $\{t_{i,j}^{l}\in\mathcal{A}:l\in\frac{1}{2}\mathbb{N}_{0},\ -l\leq i,j\leq l\}$
of matrix coefficients is an orthogonal basis for $\mathcal{H}_{h}$,
with
\[
(t_{i,j}^{l},t_{i^{\prime},j^{\prime}}^{l^{\prime}})=\delta_{l,l^{\prime}}\delta_{i,i^{\prime}}\delta_{j,j^{\prime}}q^{-2i}[2l+1]^{-1}\ .
\]
We also introduce the orthonormal basis $\xi_{i,j}^{l}:=t_{i,j}^{l}/\sqrt{q^{-2i}[2l+1]^{-1}}$.

The Haar state does not satisfy the trace property, but instead we
have $h(xy)=h(\vartheta(y)x)$ where $\vartheta(x)=k^{-2}\triangleright x\triangleleft k^{-2}$.
In particular on the generators we have
\[
\vartheta(a)=q^{2}a\ ,\quad\vartheta(b)=b\ ,\quad\vartheta(c)=c\ ,\quad\vartheta(d)=q^{-2}d\ .
\]
It follows from the theory of compact quantum groups that the Haar
state extends to a KMS state on the $C^{*}$-algebra $A$ for the
strongly continuous one-parameter group $\vartheta_{t}$, given by
\[
\vartheta_{t}(a)=q^{-2it}a\ ,\quad\vartheta_{t}(b)=b\ ,\quad\vartheta_{t}(c)=c\ ,\quad\vartheta_{t}(d)=q^{2it}d\ .
\]
This action can be analytically extended, and we recover the modular
group $\vartheta$ of the Haar state as $\vartheta=\vartheta_{i}$.
In particular, the associated modular operator $\Delta_{F}$ can be
written as $\Delta_{F}=\Delta_{L}\Delta_{R}$, where $\Delta_{L}$
and $\Delta_{R}$ are the left and right modular operators defined
by
\[
\Delta_{L}(t_{i,j}^{l})=q^{2j}t_{i,j}^{l}\ ,\qquad\Delta_{R}(t_{i,j}^{l})=q^{2i}t_{i,j}^{l}\ .
\]
These modular operators also implement one parameters groups of automorphisms,
which are given by $\sigma_{L,t}(x)=\Delta_{L}^{it}x\Delta_{L}^{-it}$
and $\sigma_{R,t}(x)=\Delta_{R}^{it}x\Delta_{R}^{-it}$. They can
be extended to complex actions, and we denote their extensions at
$z=i$ by $\sigma_{L}$ and $\sigma_{R}$. Restricted to $\mathcal{A}$,
they coincide with the left and right action of $k^{-2}$, that is
we have
\[
\sigma_{L}(x)=k^{-2}\triangleright x\ ,\qquad\sigma_{R}(x)=x\triangleleft k^{-2}\ .
\]
Finally we note that the modular group $\vartheta$ of the Haar state
can be rewritten in terms of these automorphisms as $\vartheta=\sigma_{L}\circ\sigma_{R}$.

\subsection{A decomposition}

We now consider a decomposition of the algebra $\mathcal{A}$ and
the Hilbert space $\mathcal{H}_{h}$ which has a particular geometrical
significance \cite{local-sphere}. For $n\in\mathbb{Z}$ define
\[
\mathcal{A}_{n}=\{x\in\mathcal{A}:\sigma_{L,t}(x)=q^{int}x\}\ .
\]
Then we have the decomposition $\mathcal{A}=\bigoplus_{n\in\mathbb{Z}}\mathcal{A}_{n}$.
The norm closure of $\mathcal{A}_{n}$, which we denote by $A_{n}$,
is the analogue of the space of continuous sections of the line bundle
over the sphere with winding number $n$. In particular the fixed
point algebra under the left action on $\mathcal{A}$, that is the
space $\mathcal{A}_{0}$, is isomorphic to the standard Podle\'{s}
sphere. This algebra decomposition can be extended to a Hilbert space
decomposition. If we denote by $\mathcal{H}_{n}=L^{2}(A_{n},h)$ the
GNS space corresponding to $A_{n}$, then we have
\[
\mathcal{H}_{h}=\bigoplus_{n\in\mathbb{Z}}\mathcal{H}_{n}\ .
\]

\section{The Dirac operator $D_{q}$}
\label{sec:Dirac}
We now turn our attention to the description of the quantum group
$SU_{q}(2)$ in the framework of spectral triples. In particular we
will consider the spectral triple introduced in \cite{suq2-kaad},
which is an example which fits into the framework of modular spectral
triples \cite{kaad}. We will focus our attention mainly on the Dirac
operator $D_{q}$, which is defined as
\[
D_{q}=\left(\begin{array}{cc}
(q^{-1}-q)^{-1}(q\partial_{k^{-2}}-1) & q^{-1/2}\partial_{e}\partial_{k^{-1}}\\
q^{1/2}\partial_{f}\partial_{k^{-1}} & (q^{-1}-q)^{-1}(1-q^{-1}\partial_{k^{-2}})
\end{array}\right)\ .
\]
It acts on the Hilbert space $\mathcal{H}_{h} \oplus \mathcal{H}_{h}$, where $\mathcal{H}_{h}$ is the GNS space constructed using the Haar
state in the previous section. The Dirac operator $D_{q}$ satisfies
some interesting properties, which we summarize in the following proposition
\cite{suq2-kaad}.
\begin{prop}
\label{prop:Dq}Let $D_{q}$ be the Dirac operator given above. Then:
\begin{enumerate}
\item the twisted commutator $[D_{q},x]_{\sigma_{L}}=D_{q}x-\sigma_{L}(x)D_{q}$
is bounded,
\item the twisted commutator is Lipschitz regular, that is $[|D_{q}|,x]_{\sigma_{L}}$
is bounded,
\item we have $D_{q}^{2}=\chi^{-1}\Delta_{L}^{-1}C_{q}$, where $C_{q}$
is the Casimir of $SU_{q}(2)$ and $\chi=\left(\begin{array}{cc}
q^{-1} & 0\\
0 & q
\end{array}\right)$.
\end{enumerate}
\end{prop}

It is interesting to point out the relation with the Dirac operator
for the Podle\'{s} sphere, which has been introduced in \cite{DS-pod},
see also \cite{local-sphere}. In our notation, the Hilbert space
for the spectral triple associated to the Podle\'{s} sphere is given
by $\mathcal{H}_{1}\oplus\mathcal{H}_{-1}$. This Hilbert space sits
inside $\mathcal{H}_{h}\oplus\mathcal{H}_{h}$, since we have the
decomposition $\mathcal{H}_{h}=\bigoplus_{n\in\mathbb{Z}}\mathcal{H}_{n}$.
If we restrict the Dirac operator $D_{q}$ to this subspace we obtain
\[
D_{q}|_{\mathcal{H}_{1}\oplus\mathcal{H}_{-1}}=\left(\begin{array}{cc}
(q^{-1}-q)^{-1}(qq^{-1}-1) & q^{-1/2}q^{1/2}\partial_{e}\\
q^{1/2}q^{-1/2}\partial_{f} & (q^{-1}-q)^{-1}(1-q^{-1}q)
\end{array}\right)=\left(\begin{array}{cc}
0 & \partial_{e}\\
\partial_{f} & 0
\end{array}\right)\ .
\]
Therefore it reproduces the usual Dirac operator for the Podle\'{s}
sphere, which makes it a natural object to consider. We also point
out that, since the Podle\'{s} sphere corresponds to the fixed point
algebra of $\mathcal{A}$ under the left action, it follows that the
twisted commutator condition $[D_{q},x]_{\sigma_{L}}$ reduces to
the usual one.

\subsection{The left-covariant differential calculus}

An interesting feature of the operator $D_{q}$, which has not been pointed out in \cite{suq2-kaad},
is that it implements one of the left-covariant differential calculi on $SU_{q}(2)$.
We will now show that the calculus defined by $D_{q}$ is isomorphic to the number $10$ of the list given in \cite{hecken}, where a complete classification of left-covariant differential calculi on $SU_{q}(2)$ is obtained.
This particular calculus has been considered previously, in the context of twisted spectral triples, in the paper \cite{twisted-calculi}, where it appears as an example of a more general framework.
However, the operator $D_{q}$ that we consider is slightly different from the one that appears there.

\begin{prop}
The differential calculus implemented by $D_{q}$ is isomorphic to
one of the left covariant differential calculi on $SU_{q}(2)$.\end{prop}
\begin{proof}
To prove this statement recall that two first-order differential calculi
$(\Omega_{1}^{1},d_{1})$ and $(\Omega_{2}^{1},d_{2})$ are isomorphic
if and only if $\sum_{j}a_{j}d_{1}b_{j}=0$ always implies that $\sum_{j}a_{j}d_{2}b_{j}=0$.
For a twisted spectral triple we can realize a differential calculus
in the following way: we define $\Omega_{D}^{1}$ to be the span of
operators of the form $a\cdot[D,b]_{\sigma}$ with bimodule structure
given by $a\cdot[D,b]_{\sigma}\cdot c=\sigma(a)[D,b]_{\sigma}c$,
where $a,b,c\in\mathcal{A}$. Then it is easy to check that $d_{\sigma}(a)=[D,a]_{\sigma}$
defines a derivation with values in $\Omega_{D}^{1}$, see also \cite{twisted-calculi}.

To proceed we compute the twisted commutator of $D_{q}$ with $x \in \mathcal{A}$.
A simple computation, which uses the coproduct structure of $U_{q}(\mathfrak{sl}(2))$, shows that
\begin{equation*}
\begin{split}
[D_{q},x]_{\sigma_{L}} & =
(q^{-1}-q)^{-1}
\left(\begin{array}{cc}
1 & 0\\
0 & -1
\end{array}\right)
(\partial_{k^{-2}}(x)-x)\\
& + q^{-1/2}
\left(\begin{array}{cc}
0 & 1\\
0 & 0
\end{array}\right)
\partial_{e} (\partial_{k^{-1}}(x))
+ q^{1/2}
\left(\begin{array}{cc}
0 & 0\\
1 & 0
\end{array}\right)
\partial_{f} (\partial_{k^{-1}}(x)).
\end{split}
\end{equation*}
We note in passing that this expression shows that it is a bounded operator.
Using this formula it is easy to see that the calculus defined by $D_{q}$ is isomorphic to the one given in \cite{twisted-calculi}.
This one in turn is, by construction, isomorphic to the differential calculus number $10$ in Heckenberger's list \cite{hecken}, from which the claim follows.
\end{proof}

\subsection{A regularity property}
The Dirac operator $D_{q}$ satisfies the Lipschitz regularity property,
that is $[|D_{q}|,x]_{\sigma_{L}}$ is bounded for every $x\in\mathcal{A}$,
see \cite[Lemma 3.5]{suq2-kaad}. Here we prove a similar regularity
property, namely the one that appears as a requirement in Theorem
\ref{thm:modular} for $r=1$.
\begin{lem}
\label{lem:regular}The operator $|D_{q}|[|D_{q}|^{s},x]_{\sigma^{s}}|D_{q}|^{-s}$
is bounded for every $x\in\mathcal{A}$ and for all $s\in\mathbb{R}$.\end{lem}
\begin{proof}
We start by noting that, since the Dirac operator satisfies $D_{q}^{2}=\chi^{-1}\Delta_{L}^{-1}C_{q}$,
the action of $|D_{q}|$ on the two components of the Hilbert space
$\mathcal{H}_{h} \oplus \mathcal{H}_{h}$ is the same up to a constant. Therefore we can restrict our attention to one of them, let us say the first one, on which we have $|D_{q}|\xi_{i,j}^{l}=q^{1/2}q^{-j}[l+1/2]\xi_{i,j}^{l}$.
Moreover, since the twisted commutator is well-behaved with respect
to products and adjoints, we can restrict to the case $x=a$ or $x=c$.
We can decompose the action of these operators on an element $\xi_{i,j}^{l}$
of the Hilbert space as
\[
\begin{split}a\xi_{i,j}^{l} & =\alpha_{i,j}^{l+}\xi_{i-1/2,j-1/2}^{l+1/2}+\alpha_{i,j}^{l-}\xi_{i-1/2,j-1/2}^{l-1/2}\ ,\\
c\xi_{i,j}^{l} & =\gamma_{i,j}^{l-}\xi_{i+1/2,j-1/2}^{l+1/2}+\gamma_{i,j}^{l-}\xi_{i+1/2,j-1/2}^{l-1/2}\ .
\end{split}
\]
We have the following bounds on the coefficients $\alpha_{i,j}^{l+},\gamma_{i,j}^{l+}\leq C_{1}q^{l+j}$
and $\alpha_{i,j}^{l-},\gamma_{i,j}^{l-}\leq C_{2}$, see \cite[Lemma 3.5]{suq2-kaad}.
We start by considering the case $x=a$. Then we immediately obtain
\[
\begin{split}[|D_{q}|^{s},a]_{\sigma^{s}}\xi_{i,j}^{l} & =\alpha_{i,j}^{l+}q^{s/2}(q^{-s(j-1/2)}[l+1]^{s}-q^{-s(j-1)}[l+1/2]^{s})\xi_{i-1/2,j-1/2}^{l+1/2}\\
 & +\alpha_{i,j}^{l-}q^{s/2}(q^{-s(j-1/2)}[l]^{s}-q^{-s(j-1)}[l+1/2]^{s})\xi_{i-1/2,j-1/2}^{l-1/2}\ .
\end{split}
\]

Now we want to show that $|D_{q}|[|D_{q}|^{s},a]_{\sigma^{s}}|D_{q}|^{-s}$
is a bounded operator. To do this we apply it to $\xi_{i,j}^{l}$,
compute the inner product with $\xi_{i-1/2,j-1/2}^{l\pm1/2}$ and
then show that both terms are bounded by a constant, which does not
depend on $l$. For the first one we have 
\[
\begin{split} & \quad\left(|D_{q}|[|D_{q}|^{s},a]_{\sigma^{s}}|D_{q}|^{-s}\xi_{i,j}^{l},\xi_{i-1/2,j-1/2}^{l+1/2}\right)\\
 & =q\alpha_{i,j}^{l+}q^{-j}(q^{s/2}[l+1]^{s}-q^{s}[l+1/2]^{s})[l+1/2]^{-s}[l+1]\ .
\end{split}
\]
Using the inequality $\alpha_{i,j}^{l+}\leq C_{1}q^{l+j}$ and $[l]\sim q^{-l}$,
valid for large $l$, we obtain
\[
\left(|D_{q}|[|D_{q}|^{s},a]_{\sigma^{s}}|D_{q}|^{-s}\xi_{i,j}^{l},\xi_{i-1/2,j-1/2}^{l+1/2}\right)\leq C_{1}^{\prime}q^{l+j}q^{-j}(q^{s/2}-q^{s})q^{-sl}q^{sl}q^{-l}\leq C_{1}^{\prime\prime}\ .
\]
Computing the other inner product we get
\[
\begin{split} & \quad(|D_{q}|[|D_{q}|^{s},a]_{\sigma^{s}}|D_{q}|^{-s}\xi_{i,j}^{l},\xi_{i-1/2,j-1/2}^{l-1/2})\\
 & =q\alpha_{i,j}^{l-}q^{-j}(q^{s/2}[l]^{s}-q^{s}[l+1/2]^{s})[l+1/2]^{-s}[l]\ .
\end{split}
\]
To bound this term we first observe that
\[
q^{s/2}[l]^{s}-q^{s}[l+1/2]^{s}=\frac{q^{s/2}}{(q^{-1}-q)^{s}}\left((q^{-l}-q^{l})^{s}-(q^{-l}-q^{l+1})^{s}\right)\ .
\]
Then for large $l$ we find $q^{s/2}[l]^{s}-q^{s}[l+1/2]^{s}\sim q^{-sl}q^{2l}$.
Using this result and $\alpha_{i,j}^{l+}\leq C_{2}$ we find
\[
(|D_{q}|[|D_{q}|^{s},a]_{\sigma^{s}}|D_{q}|^{-s}\xi_{i,j}^{l},\xi_{i-,j-}^{l-})\leq C_{2}^{\prime}q^{-j}q^{-sl}q^{2l}q^{sl}q^{-l}\leq C_{2}^{\prime\prime}q^{l-j}
\]
Since $-l\leq j\leq l$ we have that this term is bounded. The proof
for the case $x=c$ is completely analogous, since $\gamma_{i,j}^{\pm l}$
satisfies the same bounds as $\alpha_{i,j}^{\pm l}$, therefore we skip it.
\end{proof}

\section{The zeta function}
\label{sec:zeta}
In this section we define a family of zeta functions, depending on
the Dirac operator $D_{q}$ and a family of weights, with the aim
of studying the corresponding notion of non-commutative integration.
We point out that is is not possible to use the operator trace, since
it is known that in this case the spectral dimension does not exist
\cite{suq2-kaad}. We now wish to restrict the freedom in the choice
of the weight $\phi$ by imposing some natural conditions. In view
of the requirement that $\sigma^{\phi}(x)\in\mathcal{A}$, a natural
one is that $\Delta_{\phi}$ implements an automorphism of $SU_{q}(2)$.
The complete list of automorphisms for $SL_{q}(2)$ can be found in
\cite{twist-sl2}: there are two families, one of which acts diagonally and depends on two parameters. In the following we consider only
the diagonal case, which takes the following form on the generators
\[
\sigma_{\lambda,\mu}(a)=\lambda a\ ,\quad\sigma_{\lambda,\mu}(b)=\mu b\ ,\quad\sigma_{\lambda,\mu}(c)=\mu^{-1}c\ ,\quad\sigma_{\lambda,\mu}(d)=\lambda^{-1}d\ .
\]
We point out that the modular group $\vartheta$ of the Haar state
is of this form, with $\lambda=q^{-2}$ and $\mu=1$. Therefore we
can parametrize our weight by two real number $a,b\in\mathbb{R}$
as
\[
\phi^{(a,b)}(\cdot):=\mathrm{Tr}(\Delta_{L}^{-a}\Delta_{R}^{b}\cdot)\ .
\]
The minus sign is choosen for later convenience.

\subsection{The spectral dimension}

We start by computing the spectral dimension associated to the zeta
function constructed with $D_{q}$ and $\phi^{(a,b)}$. This imposes
some restrictions on the values of the parameters $a,b$. Moreover
we discuss the meromorphic extension of this function.
\begin{prop}
\label{prop:spec-dim}Let $\zeta^{(a,b)}(z):=\mathrm{Tr}(\Delta_{L}^{-a}\Delta_{R}^{b}|D_{q}|^{-z})$.
Then
\begin{enumerate}
\item if $a\pm b>0$ then $\zeta^{(a,b)}(z)$ is holomorphic for all $z\in\mathbb{C}$
such that $\mathrm{Re}(z)>a+|b|$,
\item in this case the corresponding spectral dimension is $n=a+|b|$,
\item $\zeta^{(a,b)}(z)$ has a meromorphic extension to the complex plane,
with only simple poles if $b\neq0$ and with only double poles if
$b=0$.
\end{enumerate}
\end{prop}
\begin{proof}
From Proposition \ref{prop:Dq} we have $D_{q}^{2}=\chi^{-1}\Delta_{L}^{-1}C_{q}$,
where $C_{q}$ is the Casimir and 
\[
\chi=\left(\begin{array}{cc}
q^{-1} & 0\\
0 & q
\end{array}\right)\ .
\]
Therefore we can write $|D_{q}|^{-z}=\chi^{z/2}\Delta_{L}^{z/2}C_{q}^{-z/2}$.
The Hilbert space is $\mathcal{H}=\mathcal{H}_{h}\oplus\mathcal{H}_{h}$,
where $\mathcal{H}_{h}$ is the GNS space constructed using the Haar
state. An orthonormal basis for this space is given by $\{\xi_{i,j}^{l}\in\mathcal{A}:l\in\frac{1}{2}\mathbb{N}_{0},\ -l\leq i,j\leq l\}$.
Then we have
\[
\zeta^{(a,b)}(z)=(q^{-z/2}+q^{z/2})\sum_{2l=0}^{\infty}\sum_{i,j=-l}^{l}(\xi_{i,j}^{l},\Delta_{L}^{-a}\Delta_{R}^{b}\Delta_{L}^{z/2}C_{q}^{-z/2}\xi_{i,j}^{l})\ .
\]
The modular operators act as $\Delta_{L}\xi_{i,j}^{l}=q^{2j}\xi_{i,j}^{l}$,
$\Delta_{R}\xi_{i,j}^{l}=q^{2i}\xi_{i,j}^{l}$, while for Casimir
we have $C_{q}\xi_{i,j}^{l}=[l+1/2]_{q}^{2}\xi_{i,j}^{l}$. Therefore
we get
\[
(\xi_{i,j}^{l},\Delta_{L}^{-a}\Delta_{R}^{b}\Delta_{L}^{z/2}C_{q}^{-z/2}\xi_{i,j}^{l})=q^{(z-2a)j}q^{2bi}[l+1/2]_{q}^{-z}\ .
\]

To proceed we use the following trick \cite{res-form}. For every
$z\in\mathbb{C}$ we have the absolutely convergent series expansion
\[
\begin{split}[l+1/2]_{q}^{-z} & =(q^{-1}-q)^{z}q^{(l+1/2)z}(1-q^{2l+1})^{-z}\\
 & =(q^{-1}-q)^{z}q^{(l+1/2)z}\sum_{k=0}^{\infty}\left(\begin{array}{c}
z+k-1\\
k
\end{array}\right)q^{(2l+1)k}\ .
\end{split}
\]
Therefore we can rewrite our zeta function as
\[
\zeta^{(a,b)}(z)=\frac{q^{-z/2}+q^{z/2}}{(q^{-1}-q)^{-z}}\sum_{2l=0}^{\infty}\sum_{i,j=-l}^{l}\sum_{k=0}^{\infty}\left(\begin{array}{c}
z+k-1\\
k
\end{array}\right)q^{(z-2a)j}q^{2bi}q^{(l+1/2)z}q^{(2l+1)k}\ .
\]
Now we consider the sum
\[
S_{k}^{(a,b)}(z):=\sum_{2l=0}^{\infty}\sum_{i,j=-l}^{l}q^{(z-2a)j}q^{2bi}q^{(l+1/2)z}q^{(2l+1)k}\ .
\]
The sums over $i$ and $j$ can be easily performed and we get
\[
S_{k}^{(a,b)}(z)=\sum_{2l=0}^{\infty}\frac{q^{(z-2a)}q^{(z-2a)l}-q^{-(z-2a)l}}{q^{(z-2a)}-1}\frac{q^{2b}q^{2bl}-q^{-2bl}}{q^{2b}-1}q^{(l+1/2)z}q^{(2l+1)k}\ .
\]
We can break this sum into four terms
\[
S_{k}^{(a,b)}(z)=\frac{q^{k+z/2}}{(1-q^{z-2a})(1-q^{2b})}\left(q^{z-2(a-b)}S_{1}-q^{z-2a}S_{2}-q^{2b}S_{3}+S_{4}\right)\ ,
\]
where we have defined
\[
\begin{split} & S_{1}:=\sum_{2l=0}^{\infty}q^{2l(z-a+b+k)}\ ,\quad S_{2}:=\sum_{2l=0}^{\infty}q^{2l(z-a-b+k)}\ ,\\
 & S_{3}:=\sum_{2l=0}^{\infty}q^{2l(a+b+k)}\ ,\quad S_{4}:=\sum_{2l=0}^{\infty}q^{2l(a-b+k)}\ .
\end{split}
\]
Since $0<q<1$, the series $\sum_{l=0}^{\infty}q^{cl}$ is absolutely
convergent when $\mathrm{Re}(c)>0$. We want this to be the case for
any $k\geq0$. From $S_{3}$ and $S_{4}$ we see that this imposes
$a+b>0$ and $a-b>0$. For $S_{1}$ and $S_{2}$, that depend on $z$,
instead we have to require
\[
\mathrm{Re}(z)>a-b\ ,\quad\mathrm{Re}(z)>a+b\ .
\]
We can then easily sum the geometric series and, after some rearranging,
we arrive at
\[
S_{k}^{(a,b)}(z)=\frac{q^{k+z/2}\left(1-q^{2k+z}\right)}{\left(1-q^{z-(a+b)+k}\right)\left(1-q^{z-(a-b)+k}\right)\left(1-q^{a+b+k}\right)\left(1-q^{a-b+k}\right)}\ .
\]
Now, going back to the expression for $\zeta^{(a,b)}(z)$, we see
that we can safely exchange the sum over $k$ with the other sums.
The result is then
\[
\zeta^{(a,b)}(z)=\frac{q^{-z/2}+q^{z/2}}{(q^{-1}-q)^{-z}}\sum_{k=0}^{\infty}\left(\begin{array}{c}
z+k-1\\
k
\end{array}\right)S_{k}^{(a,b)}(z)\ .
\]
The statement about the meromorphic extension is clear from the form
of $S_{k}^{(a,b)}(z)$.
\end{proof}
In the following we will assume that the conditions $a\pm b>0$ are
satisfied, in such a way that the spectral dimension exists. Moreover
we exclude the case $b=0$, since in this case the zeta function has
a double pole at the spectral dimension.

\subsection{The modular property}
We now consider the non-commutative integral $\varphi$ associated
to the zeta function and determine its modular group $\theta$. In
particular we can investigate the connection with the Haar state of
$SU_{q}(2)$, which satisfies the property $h(xy)=h(\vartheta(y)x)$,
where $\vartheta=\sigma_{L}\circ\sigma_{R}$. Therefore, to recover
the Haar state from the non-commutative integral, a necessary condition
is that $\theta=\vartheta$. We now show that this condition fixes
the parameter $b$ to be equal to one.
\begin{prop}
Let $\zeta_{x}^{(a,b)}(z)=\mathrm{Tr}(\Delta_{L}^{-a}\Delta_{R}^{b}x|D_{q}|^{-z})$,
with $n=a+|b|$ be the associated spectral dimension.
Let $\theta$ be the modular grup of the non-commutative integral and $\vartheta=\sigma_{L}\circ\sigma_{R}$
the modular group of the Haar state. Then we have $\theta=\vartheta$
if and only if $b=1$.\end{prop}
\begin{proof}
We can apply Theorem \ref{thm:modular} to the non-commutative integral
$\varphi$. Indeed by Lemma \ref{lem:regular} we have that $|D|[|D|^{s},y]_{\sigma^{s}}|D|^{-s}$
is bounded for every $s\in\mathbb{R}$, while by Proposition \ref{prop:Dq}
the twist in the commutator $\sigma=\sigma_{L}$ acts diagonally on
$\mathcal{A}$. Therefore we obtain $\varphi(xy)=\varphi(\theta(y)x)$,
with $\theta=\sigma^{\phi}\circ\sigma^{n}$. In the case under consideration
we have $\sigma_{\phi}=\sigma_{L}^{-a}\circ\sigma_{R}^{b}$, so we
get
\[
\theta=\sigma_{\phi}\circ\sigma^{n}=\sigma_{L}^{n-a}\circ\sigma_{R}^{b}=\sigma_{L}^{|b|}\circ\sigma_{R}^{b}\ ,
\]
where we have used the fact that the spectral dimension is given by
$n=a+|b|$.

For $b<0$ we have $\theta=\sigma_{L}^{|b|}\circ\sigma_{R}^{-|b|}$
and it is clear that there is no solution. On the other hand for $b>0$
we have $\theta=\sigma_{L}^{b}\circ\sigma_{R}^{b}$, so the equality
$\theta=\vartheta$ holds for $b=1$.
\end{proof}
This result shows that we can only partially fix the arbitrariness
in the choice of the weight $\phi^{(a,b)}(\cdot)=\mathrm{Tr}(\Delta_{L}^{-a}\Delta_{R}^{b}\cdot)$.
Indeed, as we have seen in the proof given above, the dependence on
the parameter $a$ cancels in the combination $\theta=\sigma_{\phi}\circ\sigma^{n}$.
It is worth pointing out that a similar phenomenon happens also in
the spectral triple considered in \cite{kmink1,kmink2}, where similar
techniques are employed. This is expected to happen, quite generically,
when the twist in the commutator also appears in the modular group
of the weight.

Of course a natural requirement to fix this ambiguity would be to recover
the classical dimension, which would fix $a=2$. In the last part
of the paper we will consider another condition, more spectral in
nature, which also fixes uniquely $a=2$. It is also of some interest
to remark that, if one requires $n$ to be an integer, then the smallest
$n$ which is allowed by the previous analysis is indeed $n=3$. Finally,
for examples coming from quantum groups, this ambiguity in the choice
of the weight could be related to a similar one that arises in twisted
Hochschild homology: indeed it is known that a twist is necessary
to avoid the dimension drop, but it happens that one finds a family
of such twists, see for example \cite{twist-sl2}.

Let us also mention what happens for the Podle\'{s} sphere. In this
case, since the Hilbert space is given by $\mathcal{H}_{1}\oplus\mathcal{H}_{-1}$,
the modular operator $\Delta_{L}^{a}$ gives a constant matrix, which
can be absorbed in the normalization. Therefore it does not affect
the spectral dimension and the modular group of the non-commutative
integral. As we mentioned before, the twist in the commutator disappears,
since the Podle\'{s} sphere is the fixed point algebra of $\mathcal{A}$
under the left action. Then it is easy to repeat the previous analysis,
with the result that we must fix the value $b=1$ if we want to recover
the modular group of the Haar state. Moreover it follows from the
results of \cite{res-form} that the corresponding spectral dimension
is $n=2$. Therefore our results for $SU_{q}(2)$ restrict in a natural
way to the case of the Podle\'{s} sphere.

\section{The Haar state}
\label{sec:Haar}
So far we have only shown that the non-commutative integral $\varphi$
has the same modular group of the Haar state $h$, which leaves open
the question of whether they are equal. In principle this could happen
for some values of $a$, or maybe for none at all. In this section
we show that the non-commutative integral coincides with the Haar
state for all allowed values of $a$. Of course we must normalize
$\varphi$, since the Haar state satisfies $h(1)=1$, while in general
we do not have $\varphi(1)=1$. This normalization is achieved by
computing $\varphi(1)$, that is the residue of $\zeta^{(a,1)}(z)$
at the spectral dimension $n=a+1$. The result of this computation
is
$$
\varphi(1)= \res{z=a+1}\ \zeta^{(a,1)}(z)=\frac{(q^{-1}-q)^{a}(q^{a+1}+1)}{(q^{a}-q)\log(q)}.
$$
We denote the normalized non-commutative integral as $\tilde{\varphi}(x):=\varphi(x)/\varphi(1)$.
Notice that the normalization $\varphi(1)$ depends on $a$. On the
other hand we will now show that $\tilde{\varphi}(x)$ is independent
of $a$ and recovers the Haar state.

\subsection{Approximating the GNS representation}

To proceed with the computation of the non-commutative integral it is convenient to work with a different representation of $SU_{q}(2)$.
This representation, which we denote by $\rho$, approximates the GNS representation, as we shall see in the next lemma. It is defined
on the generators as
\[
\begin{split}\rho(a)\xi_{i,j}^{l} & :=\sqrt{1-q^{2(l+i)}}\xi_{i-1/2,j-1/2}^{l-1/2}\ ,\\
\rho(b)\xi_{i,j}^{l} & :=-q^{l+i+1}\xi_{i-1/2,j+1/2}^{l+1/2}\ ,\\
\rho(c)\xi_{i,j}^{l} & :=q^{l+i}\xi_{i+1/2,j-1/2}^{l-1/2}\ ,\\
\rho(d)\xi_{i,j}^{l} & :=\sqrt{1-q^{2(l+i+1)}}\xi_{i+1/2,j+1/2}^{l+1/2}\ .
\end{split}
\]
One can easily check that $\rho$ is a representation of $SU_{q}(2)$,
as it satisfies its defining relations. For more details see \cite[Proposition 9.4]{kaad}
and references therein.
\begin{lem}
For any $x\in\mathcal{A}$ we have the equality
\[
\varphi(x) = \res{z=n}\ \mathrm{Tr}(\Delta_{L}^{-a}\Delta_{R}\rho(x)|D_{q}|^{-z})\ .
\]
\end{lem}
\begin{proof}
We need to show that $\Delta_{L}^{-a}\Delta_{R}(x-\rho(x))|D_{q}|^{-z}$
is trace-class for $z=n$, that is
\[
\mathrm{Tr}\left(|\Delta_{L}^{-a}\Delta_{R}(x-\rho(x))|D_{q}|^{-n}|\right)<\infty\ .
\]
Using the fact that $xy-\rho(xy)=(x-\rho(x))\rho(y)+x(y-\rho(y))$,
we can restrict our attention to the generators. Moreover, since $\rho$
is a $*$-representation of $SU_{q}(2)$, it suffices to consider
the cases $x=a$ and $x=c$. Recall that for the GNS representation
we have the formulae
\[
\begin{split}a\xi_{i,j}^{l} & =\alpha_{i,j}^{l+}\xi_{i-1/2,j-1/2}^{l+1/2}+\alpha_{i,j}^{l-}\xi_{i-1/2,j-1/2}^{l-1/2}\ ,\\
c\xi_{i,j}^{l} & =\gamma_{i,j}^{l+}\xi_{i+1/2,j-1/2}^{l+1/2}+\gamma_{i,j}^{l-}\xi_{i+1/2,j-1/2}^{l-1/2}\ .
\end{split}
\]

To prove that the trace is finite, in the case $x=a$, we need to
estimate the quantities $|\alpha_{i,j}^{l+}|$ and $\left|\alpha_{i,j}^{l-}-\sqrt{1-q^{2(l+i)}}\right|$.
Recall that we already used the fact that $|\alpha_{i,j}^{l+}|\leq C_{+}q^{l+j}$.
Similarly one proves that $\left|\alpha_{i,j}^{l-}-\sqrt{1-q^{2(l+i)}}\right|\leq C_{-}q^{l+j}$,
see \cite[Proposition 9.4]{kaad}. Therefore we only need to repeat
the computation of the sum $\xi^{(a,1)}(z)$ with the factor $q^{l+j}$
inserted. We can easily perform this sum as we did in the computation
of the spectral dimension. The result is finite for $z=n$, so that
the operator $\Delta_{L}^{-a}\Delta_{R}(x-\rho(x))|D_{q}|^{-n}$ is
trace-class.

The computation is completely identical for the case $x=c$, since
we have the estimates $|\gamma_{i,j}^{l+}|\leq C_{+}^{\prime}q^{l+j}$
and $\left|\gamma_{i,j}^{l-}-q^{l+i}\right|\leq C_{-}^{\prime}q^{l+j}$.
Therefore the equality is proven.
\end{proof}

\subsection{The computation}

Using the representation $\rho$, we can now easily compute $\tilde{\varphi}(x)$
for any $x\in\mathcal{A}$. We show that this non-commutative
integral coincides with $h(x)$, where $h$ is the Haar state, independently of the value of the parameter $a$ which appears in the definition of $\tilde{\varphi}$.
\begin{thm}
For any $x\in\mathcal{A}$ we have $\tilde{\varphi}(x)=h(x)$, where
$h$ is the Haar state.\end{thm}
\begin{proof}
Recall that the Haar state $h$ takes the following values on the
generators
\[
h(a^{i}b^{j}c^{k})=h(d^{i}b^{j}c^{k})=\delta_{i,0}\delta_{j,k}(-1)^{k}[k+1]_{q}^{-1}\ .
\]
Using the previous lemma and the explicit formulae for the approximate
representation, it is not difficult to see that the non-commutative
integral must have the following form 
\[
\tilde{\varphi}(a^{i}b^{j}c^{k})=\tilde{\varphi}(d^{i}b^{j}c^{k})=\delta_{i,0}\delta_{j,k}\tilde{\varphi}(b^{j}c^{k})\ .
\]
Therefore, to prove that $\tilde{\varphi}$ is the Haar state, it
only remains to show that $\tilde{\varphi}(b^{n}c^{n})$ coincides
with $h(b^{n}c^{n})$. Using the representation $\rho$ we can immediately
compute
\[
\rho(b^{n}c^{n})\xi_{i,j}^{l}=(-1)^{n}q^{n(l+i+1)}q^{n(l+i)}\xi_{i,j}^{l}=(-1)^{n}q^{2nl}q^{2ni}q^{n}\xi_{i,j}^{l}\ .
\]

Now we only need to repeat the computation of Proposition \ref{prop:spec-dim}
by inserting the factor $(-1)^{n}q^{2nl}q^{2ni}q^{n}$. We omit this
computation. The result is that there is a simple pole at $z=a+1$,
that is the spectral dimension, whose residue is non-zero. Explicitely
we obtain
\[
\tilde{\varphi}(b^{n}c^{n})=(-1)^{n}q^{n}\frac{q^{2}-1}{q^{2(n+1)}-1}=(-1)^{n}[n+1]_{q}^{-1}\ .
\]
This coincides with $h(b^{n}c^{n})$, so the proof is complete.
\end{proof}

\section{An additional requirement}
\label{sec:scalar}
In the previous section we have shown that the non-commutative integral
coincides with the Haar state, regardless of the value of the parameter $a$.
Of course, as we mentioned before,
since the spectral dimension is given by $n=a+1$, we have the natural choice $a=2$ which gives the classical dimension.
In this section we propose a different criterion to fix this free parameter, which is based on the value of a certain residue of the zeta function.
It will turn out that this criterion is satisfied only for $a=2$.

To formulate this criterion we note that, by requiring the non-commutative integral
to be equal to the Haar state, we have imposed a condition on the
residue at $z=n$ of the zeta function. But the zeta function contains
much more information than this residue: indeed in the classical
case we know, for example from the heat kernel expansion, that also the other residues contain geometrical information, like the scalar curvature and various contractions of the Riemann tensor.
We can therefore look at these other residues to impose an additional requirement. In particular
we can look at the next non-trivial coefficient of the expansion.

For this reason we briefly recall how the heat kernel expansion works,
and how we can use it for our needs. Let $M$ be a compact Riemannian
manifold of dimension $n$ with a fixed metric $g$. Consider a second
order operator of Laplace-type, which locally can be written as
\begin{equation}
P=-(g^{\mu\nu}\nabla_{\mu}\nabla_{\nu}+E)\ .\label{eq:form}
\end{equation}
For any smooth function $f$ on $M$ we can consider the operator
$f\exp(-tP)$, for $t>0$. Then there is an asymptotic expansion of $\mathrm{Tr}(f\exp(-tP))$, for $t\downarrow0$, which is given by
\[
\mathrm{Tr}(fe^{-tP})\sim\sum_{k=0}^{\infty}t^{(k-n)/2}a_{k}(f,P)\ ,
\]
where the coefficients $a_{k}(f,P)$ can be expressed as integrals
of local invariants of $M$. For a manifold without boundary only
the even coefficients are non-zero. In the following we will consider
only the first two non-zero coefficients, which read as follows
\[
a_{0}(f,P)=(4\pi)^{-n/2}\int_{M}f\sqrt{g}d^{n}x\ ,\qquad a_{2}(f,P)=(4\pi)^{-n/2}6^{-1}\int_{M}f(6E+R)\sqrt{g}d^{n}x\ .
\]
Here $R$ is the scalar curvature associated to the metric $g$.

The coefficients of the heat kernel expansion are closely related
to the residues of the zeta function. Indeed, in the case of a positive
$P$, consider the zeta function defined as $\zeta(z,f,P)=\mathrm{Tr}(fP^{-z})$.
Then the heat kernel coefficients $a_{k}(f,P)$ are given by 
\[
a_{k}(f,P) = \res{z=(n-k)/2}\Gamma(z)\zeta(z,f,P).
\]
We remark that, from this relation, we obtain yet another justification for our definition of the non-commutative integral. Indeed, for the zeta function of a first order operator, like the Dirac operator $D$, the residue at the spectral dimension $n$ is proportional to the coefficient $a_{0}(f,P)$, which as recalled above is proportional to
the integral of $f$.

We now want to look at the next non-trivial coefficient of the expansion,
which is given by $a_{2}(f,P)$. This corresponds, for the zeta function
defined in terms of the Dirac operator, to the residue at $z=n-2$.
We first compute this coefficient for the classical limit of our Dirac
operator $D_{q}$. In the following we set $f=1$ and we drop the
dependence on $f$ in the notation.

\subsection{The commutative limit of $D_{q}$}

The manifold corresponding to the group $SU(2)$ is the $3$-sphere.
On this space we consider the Laplace-Beltrami operator $\Delta=-g^{\mu\nu}\nabla_{\mu}\nabla_{\nu}$,
which has eigenvalues $k(k+2)$ with multiplicity $(k+1)^{2}$, where
$k\in\mathbb{N}_{0}$. It is not hard to show that $a_{2}(\Delta)=\sqrt{\pi}/4$.
Since for the operator $\Delta$ we have from (\ref{eq:form}) that
$E=0$, it follows that the scalar curvature is $R=6$.

In the non-commutative case the operator $D_{q}$ satisfies $D_{q}^{2}=\chi^{-1}\Delta_{L}^{-1}C_{q}$,
where $\chi$ is a constant matrix, $\Delta_{L}$ acts as $\Delta_{L}\xi_{i,j}^{l}=q^{2j}\xi_{i,j}^{l}$
and $C_{q}$ is the Casimir of $SU_{q}(2)$. In particular $C_{q}$
has the eigenvalues $[l+1/2]_{q}^{2}$ with multiplicity $(2l+1)^{2}$,
where $l=\frac{1}{2}\mathbb{N}_{0}$. Now in the commutative limit
$q\to1$ the matrix $\chi$ reduces to the identity, $\Delta_{L}$
reduces to the identity operator and the eigenvalues of $C_{q}$ become
$(l+1/2)^{2}$. Therefore we see that, upon writing $k=l/2$, we reduce
to the classical situation of an operator $C$ with eigenvalues $\frac{1}{4}(k+1)^{2}$
and multiplicity $(k+1)^{2}$, where $k\in\mathbb{N}_{0}$. It is
clear that $C$ is related to the Laplace-Beltrami operator $\Delta$
by a rescaling and the addition of a costant, that is $C=\frac{1}{4}\Delta+\frac{1}{4}$.

Therefore we can compare $C_{q}$ with its classical limit $C$. To
this end we look at the heat kernel coefficients of the operator $C$,
specifically at $a_{2}(C)$. We can easily obtain it from the knowledge
of $a_{2}(\Delta)$ in the following way: the rescaling $\Delta\to\frac{1}{4}\Delta$
has the effect a conformal transformation $g\to4g$ of the metric,
which in turn changes the scalar curvature by $R\to\frac{1}{4}R$.
Then the addition of the constant $\frac{1}{4}$ simply sets $E=-\frac{1}{4}$.
Therefore we obtain
\[
6E+R\to-\frac{6}{4}+\frac{R}{4}=0\ ,
\]
where we have used the fact that for the $3$-sphere the scalar curvature
is $R=6$. In other words we have that $a_{2}(C)=0$. Another way
to check that this is the case is directly via the zeta function.
Indeed, after removing the zero eigenvalue, it is simple to compute
\[
\zeta(z,C)=\sum_{k=1}^{\infty}(k+1)^{2}4^{-z}(k+1)^{-2z}=4^{-z}(\zeta(2z-2)-1)\ .
\]
This function is regular at $z=(3-2)/2=1/2$, so that we have $a_{2}(C)=0$.

\subsection{The requirement on the residue}
Let us now get back to our original problem. We have seen that the
non-commutative integral recovers the Haar state, independently of
the value of $a$. But now, from the previous discussion, we have a natural requirement that could possibly fix this ambiguity: since, as we
have seen, in the non-commutative case the role of $C$ is played
by $D_{q}^{2}$, we can try to impose the analogue of the condition
$a_{2}(C)=0$. This means that we can require
$$
\res{z=n-2}\Gamma(z)\zeta^{(a,1)}(z) = 0.
$$
Recall that the spectral dimension $n$ depends on $a$, since $n=a+1$.
The next proposition shows that this fixes the natural value $a=2$,
corresponding to the classical dimension.
\begin{prop}
The residue of $\zeta^{(a,1)}(z)$ at $z=n-2$ is zero if and only
if $a=2$.\end{prop}
\begin{proof}
From the proof of Proposition \ref{prop:spec-dim} we have that
\[
\zeta^{(a,1)}(z)=\frac{q^{-z/2}+q^{z/2}}{(q^{-1}-q)^{-z}}\sum_{k=0}^{\infty}\left(\begin{array}{c}
z+k-1\\
k
\end{array}\right)S_{k}^{(a,1)}(z)\ ,
\]
where $S_{k}^{(a,1)}(z)$ is given by 
\[
S_{k}^{(a,1)}(z)=\frac{q^{k+z/2}\left(1-q^{2k+z}\right)}{\left(1-q^{z-(a+1)+k}\right)\left(1-q^{z-(a-1)+k}\right)\left(1-q^{a+1+k}\right)\left(1-q^{a-1+k}\right)}\ .
\]
Since the spectral dimension is given by $n=a+1$, we should take
the residue at $z=n-2=a-1$. Notice that for this value of $z$ the
zeta function $\zeta^{(a,1)}(z)$ has two poles, coming respectively
from the terms $k=0$ and $k=2$. Omitting a common prefactor, these
are given by
\[
\begin{split} & \quad\frac{q^{z/2}\left(1-q^{z}\right)}{\left(1-q^{z-a-1}\right)\left(1-q^{z-a+1}\right)\left(1-q^{a+1}\right)\left(1-q^{a-1}\right)}\\
 & +\frac{1}{2}z(z+1)\frac{q^{2+z/2}\left(1-q^{4+z}\right)}{\left(1-q^{z-a+1}\right)\left(1-q^{z-a+3}\right)\left(1-q^{a+3}\right)\left(1-q^{a+1}\right)}\ .
\end{split}
\]

Taking now the limit $z\to a-1$ in the regular terms we get
\[
\begin{split} & \quad\frac{q^{(a-1)/2}\left(1-q^{a-1}\right)}{\left(1-q^{-2}\right)\left(1-q^{z-a+1}\right)\left(1-q^{a+1}\right)\left(1-q^{a-1}\right)}\\
 & +\frac{1}{2}a(a-1)\frac{q^{2+(a-1)/2}\left(1-q^{a+3}\right)}{\left(1-q^{z-a+1}\right)\left(1-q^{2}\right)\left(1-q^{a+3}\right)\left(1-q^{a+1}\right)}\ .
\end{split}
\]
This expression can be rearranged as
\[
\begin{split} & \quad\frac{q^{(a-1)/2}}{1-q^{a+1}}\left(\frac{1}{1-q^{-2}}+\frac{1}{2}a(a-1)\frac{q^{2}}{1-q^{2}}\right)\frac{1}{1-q^{z-a+1}}\\
 & =\frac{q^{(a-1)/2}}{1-q^{a+1}}\frac{q^{2}}{1-q^{2}}\left(-1+\frac{1}{2}a(a-1)\right)\frac{1}{1-q^{z-a+1}}\ .
\end{split}
\]
The residue of this term is then non-zero unless we have $\frac{1}{2}a(a-1)=1$,
whose solutions are $a=-1$ and $a=2$. But we known from Proposition
\ref{prop:spec-dim} that we have to impose the conditions $a\pm1>0$
for the spectral dimension to exists. This excludes the case $a=-1$.\end{proof}

{\footnotesize
\emph{Acknowledgements}. I wish to thank Jens Kaad for helpful comments on this work. I also want to thank Francesca Arici for comments on the first version of this paper.
}

\end{document}